\documentclass[runningheads]{llncs}
\usepackage{graphicx}
\usepackage{multirow}
\usepackage{latexsym}
\usepackage{amsfonts}
\usepackage{amssymb}
\usepackage{amsmath}
\usepackage{graphicx}
\usepackage{subfigure}
\usepackage{floatrow}
\usepackage{caption}
\usepackage{color}
\usepackage{wrapfig,lipsum,booktabs}
\usepackage{semantic}
\usepackage{url}
\usepackage{tikz}
\usepackage{mathdots}
\usepackage{yhmath}
\usepackage{cancel}
\usepackage{siunitx}
\usepackage{array}
\usepackage{gensymb}
\usepackage{tabularx}
\usepackage{booktabs}
\usetikzlibrary{fadings}
\usepackage{pifont}
\usepackage{floatrow}
\usepackage[ruled,vlined]{algorithm2e}
\usepackage{rotating}
\usepackage{mathtools}
\usepackage{soul}

\DeclarePairedDelimiter{\ceil}{\lceil}{\rceil}
\newcommand{\limp}[0]{\rightarrow}
\newcommand{\N}{\mathbb{N}}

\DeclareMathOperator*{\minimise}{minimise}
\newtheorem{param}{Parameter}

\newcommand{\nodefilter}[0]{\theta}
\newcommand{\rulefilter}[0]{\phi}
\newcommand{\leafmerger}[0]{\psi}
\newcommand{\sizefilter}[0]{k}
\newcommand{\accweight}[0]{\coef}

\newcommand{\simp}[0]{Simp}

\newcommand{\fstmeasure}[0]{N}
\newcommand{\sndmeasure}[0]{L}
\newcommand{\thdmeasure}[0]{M}
\newcommand{\scoreOpt}[0]{S_{opt}}
\newcommand{\scorePro}[0]{S_{pro}}
\newcommand{\coef}[0]{\epsilon}
\newcommand{\framename}[0]{OptExplain}
\newcommand{\proname}[0]{ProClass}

\graphicspath{./images/}

\begin{document}
\title{Extracting Optimal Explanations for Ensemble Trees via Logical Reasoning
%\thanks{Supported by organization x.}
}
\author{
Gelin Zhang\inst{1}
\and
Zh\'e H\'ou\inst{2}
% \orcidID{1111-2222-3333-4444} 
\and
Yanhong Huang\inst{1}
\and
Jianqi Shi\inst{1}
\and
Hadrien Bride\inst{2}
\and
Jin Song Dong\inst{2,3}
\and
Yongsheng Gao\inst{2}
% \orcidID{2222--3333-4444-5555}
}
% %
\authorrunning{Zhang, H\'ou, Huang, Shi, Bride, Dong and Gao.}
\institute{
National Trusted Embedded Software Engineering Technology Research Center, East China Normal University, China \and
%Institute for Integrated and Intelligent Systems, 
Griffith University, Australia
%\email{z.hou@griffith.edu.au}
\and
%School of Computing, 
National University of Singapore, Singapore
%\email{dcsdjs@nus.edu.sg}
% \and
% Dependable Intelligence, Australia
}
\maketitle              % typeset the header of the contribution
\begin{abstract}
% Our approach converts an ensemble trees model to a set of decision rules and extracts an optimised explanation with high fidelity. 
Ensemble trees are a popular machine learning model which often yields high prediction performance when analysing structured data. Although individual small decision trees are deemed explainable by nature, an ensemble of large trees is often difficult to understand. In this work, we propose an approach called optimised explanation (\framename) that faithfully extracts global explanations of ensemble trees using a combination of logical reasoning, sampling and optimisation. 
Building on top of this, we propose a method called the profile of equivalent classes (\proname), which uses MAX-SAT to simplify the explanation even further\footnote[1]{The code is available at \url{https://github.com/GreeenZhang/OptExplain.}}. Our experimental study on several datasets shows that our approach can provide high-quality explanations to large ensemble trees models, and it betters recent top-performers.

\end{abstract}

\section{Introduction}
\label{sec:intro}

% Z: a bit of motivation. Should be extended.
\paragraph{Background.}
Ensemble trees are a family of machine learning techniques that combine individual decision trees to form a better prediction model. Examples include random forest~\cite{ho1995,Breiman2001}, which combines strong learners (e.g., large trees) to reduce variance and avoid overfitting. Boosting~\cite{Freund1999,friedman2002}, on the other hand, combines weak learners (e.g., small trees) to reduce bias. Ensemble trees are very successful in today's data analytics competitions and applications; they are especially suited to analyse \emph{structured data} such as databases and spreadsheets, where they sometimes outperform deep learning~\cite{szilardbenchmark}. 

Although decision trees are often deemed an explainable, or even a ``white-box'' model, such an impression usually refers to a single, short decision tree. In the context of ensemble trees such as the models generated by random forest or boosting, there can be a large number of trees and each tree can be gigantic. For example, to achieve a 0.76+ area-under-the-curve (AUC) for the 1 million flight dataset~\cite{szilardbenchmark}, Silas~\cite{BrideDDH18,silas2018} trains a model of 100 trees and each tree has more than 32,000 branches. Such a model certainly does not manifest itself in an explainable manner to the general user. The main goal of this work is to extract faithful explanations for such large-scale models.

There are several existing methods for analysing and interpreting machine learning models. For example, the LIME tool~\cite{lime2016} and the SHAP values~\cite{Lundberg2017} are both promising techniques for solving this problem. We will discuss more details of related work in Section~\ref{sec:related}. However, most existing work is done from a statistics perspective. Such methods use a prediction model as a black-box and attempt to find statistical (e.g., linear) approximations of the model. By contrast, our philosophy is that we should analyse the internal working of the model and obtain an understanding of how it works logically. Further, many existing techniques are focused on \emph{local explanations}, that is, how the model predicts for a particular data instance. %Although this is part of our interest as well (see Section~\ref{sec:instance_level_explanation}), 
This work is primarily about \emph{global explanations}, which explains how the model behaves generally. Part of the reasons why we choose ensemble trees is that decision trees are no strangers to logicians. For example, binary decision diagrams, which have a similar form, are widely used in theorem proving~\cite{Gore2014} and model checking~\cite{SLD09}. The tree structure is well-understood in the logic and verification community, and there are many possibilities to apply logical reasoning to analysing ensemble trees.

\paragraph{Our approach.}
In our previous work~\cite{silas2018}, we have used sampling and maximum satisfiable subset to extract the decision logic of the model. However, it is non-trivial to manually adjust the sampling parameters, which may lead to vastly different explanations. Default parameters often lead to very simple explanations that diverge from the original model. In this paper, we propose an integrated and automated framework for providing 
%\st{both local and } 
global explanations. Moreover, our goal is not just to give \emph{an} explanation as is done in the literature, but to give \emph{the optimal} explanation in terms of \emph{simplicity} and \emph{faithfulness}. 

An outline of our approach follows: we extract logical formulae from a set of trees where each branch forms a ``decision rule''. We then reduce the size of the model by filtering out low-quality nodes (i.e., sub-formulae) and branches. We also devise a customised formulae simplification algorithm to obtain logically equivalent smaller models. In case there are still too many decision rules, we group the rules into ``equivalent classes'' to further abstract the model. The parameters in the above process are optimised using particle swarm towards a sweet spot of simplicity and faithfulness. As an extra step, we can simplify the explanation using MAX-SAT to obtain even more abstract representations of each equivalent class, which we call the ``profiles of classes''. Such profiles can provide straightforward and even visual explanations of the model.

The utilities of this work are manifold. First, our approach can provide human-understandable explanations that are very close to the original ensemble trees model in predictive behaviour. Second, such an explanation can also be used as an approximation of the original model. For example, verification of machine learning models is another popular topic, but a general sound and complete verification algorithm for ensemble trees have proven impractical~\cite{Tornblom2019}. As a step back, we can look at the software testing scenario: since the explanation mimics the behaviour of the original model, if it violates a property, then it is likely that the original model would fail the verification, too. In such cases, we can use the explanation to constrain the search space when finding counterexamples. Finally, this work can serve as a stepping stone towards explaining deep learning. There are existing methods for converting neural networks to decision trees~\cite{Hinton2017} exactly for explanation purposes. However, these methods only induce a single decision tree, whose predictive performance is incomparable to the deep learning model. One may convert neural networks to a set of decision trees instead, then this work can be directly applied to obtain explanations. 

\paragraph{Contributions.}
The main contributions of this paper are as follows: 
\begin{enumerate}
    \item We formalise ensemble trees into logical formulae and develop simplification and abstraction algorithms that are specialised for machine learning.
    \item We propose an automated explanation extraction method called \emph{\framename}, which combines logical reasoning, sampling, and bio-inspired optimisation.
    \item We also develop a method called \emph{\proname} that computes the abstractions of each (equivalent) class using MAX-SAT.
    \item Through case studies and experiment, we show that our method is practical and useful on different datasets. It also outperforms similar tools.
\end{enumerate}

The remainder of this paper is organised as follows: Section~\ref{sec:pre} describes the preliminary concepts, Section~\ref{sec:method} details the proposed approach, Section~\ref{sec:case_studies} gives case studies and experiment, Section~\ref{sec:related} discusses related work, and Section~\ref{sec:conc} concludes the paper.

%We shall use Silas~\cite{Bride2019} as a case study for our approach, which can also be implemented for other tools such as scikit-learn~\cite{scikit-learn}, H2O~\cite{h2o}, etc.
\section{Preliminaries}
\label{sec:pre}

In this section, we redefine decision trees and their ensembles from a logical language point of view. 

\subsection{Decision Trees With a Logical Foundation}
\label{subsec:trees}

\begin{wrapfigure}{r}{0.4\textwidth}
%\begin{figure}[ht!]
  \begin{center}
\tikzset{every picture/.style={line width=0.75pt}} %set default line width to 0.75pt
\vspace{-25px}
\begin{tikzpicture}[x=0.75pt,y=0.75pt,yscale=-1,xscale=1]
%Flowchart: Decision [id:dp7258101286861609]
\draw   (367.47,166) -- (386.19,176.88) -- (367.47,187.77) -- (348.75,176.88) -- cycle ;
%Shape: Ellipse [id:dp855533740318804]
\draw   (312.38,226.41) .. controls (312.38,220.4) and (320.76,215.53) .. (331.1,215.53) .. controls (341.44,215.53) and (349.82,220.4) .. (349.82,226.41) .. controls (349.82,232.42) and (341.44,237.3) .. (331.1,237.3) .. controls (320.76,237.3) and (312.38,232.42) .. (312.38,226.41) -- cycle ;
%Straight Lines [id:da9778603522004762]
\draw    (367.47,187.77) -- (332.69,214.31) ;
\draw [shift={(331.1,215.53)}, rotate = 322.65] [color={rgb, 255:red, 0; green, 0; blue, 0 }  ][line width=0.75]    (10.93,-3.29) .. controls (6.95,-1.4) and (3.31,-0.3) .. (0,0) .. controls (3.31,0.3) and (6.95,1.4) .. (10.93,3.29)   ;
%Straight Lines [id:da39397183183371376]
\draw    (367.47,187.77) -- (401.75,215.36) ;
\draw [shift={(403.3,216.62)}, rotate = 218.82999999999998] [color={rgb, 255:red, 0; green, 0; blue, 0 }  ][line width=0.75]    (10.93,-3.29) .. controls (6.95,-1.4) and (3.31,-0.3) .. (0,0) .. controls (3.31,0.3) and (6.95,1.4) .. (10.93,3.29)   ;
%Flowchart: Decision [id:dp019267594115989883]
\draw   (403.3,216.62) -- (422.03,227.5) -- (403.3,238.38) -- (384.58,227.5) -- cycle ;
%Straight Lines [id:da7621886747750442]
\draw    (403.3,238.38) -- (437.58,265.98) ;
\draw [shift={(439.14,267.23)}, rotate = 218.82999999999998] [color={rgb, 255:red, 0; green, 0; blue, 0 }  ][line width=0.75]    (10.93,-3.29) .. controls (6.95,-1.4) and (3.31,-0.3) .. (0,0) .. controls (3.31,0.3) and (6.95,1.4) .. (10.93,3.29)   ;
%Straight Lines [id:da4719531001775997]
\draw    (403.3,238.38) -- (368.52,264.93) ;
\draw [shift={(366.93,266.14)}, rotate = 322.65] [color={rgb, 255:red, 0; green, 0; blue, 0 }  ][line width=0.75]    (10.93,-3.29) .. controls (6.95,-1.4) and (3.31,-0.3) .. (0,0) .. controls (3.31,0.3) and (6.95,1.4) .. (10.93,3.29)   ;
%Shape: Ellipse [id:dp4878828500890685]
\draw   (348.21,277.03) .. controls (348.21,271.01) and (356.59,266.14) .. (366.93,266.14) .. controls (377.27,266.14) and (385.65,271.01) .. (385.65,277.03) .. controls (385.65,283.04) and (377.27,287.91) .. (366.93,287.91) .. controls (356.59,287.91) and (348.21,283.04) .. (348.21,277.03) -- cycle ;
%Shape: Ellipse [id:dp2939149529760524]
\draw   (420.42,278.12) .. controls (420.42,272.1) and (428.8,267.23) .. (439.14,267.23) .. controls (449.48,267.23) and (457.86,272.1) .. (457.86,278.12) .. controls (457.86,284.13) and (449.48,289) .. (439.14,289) .. controls (428.8,289) and (420.42,284.13) .. (420.42,278.12) -- cycle ;
% Text Node
\draw (367.47,176.88) node  [align=left] {{\fontfamily{ptm}\selectfont $F_1$}};
% Text Node
\draw (403.3,227.5) node  [align=left] {{\fontfamily{ptm}\selectfont $F_2$}};
% Text Node
\draw (331.1,226.41) node  [align=left] {{\fontfamily{ptm}\selectfont {\footnotesize (0,6)}}};
% Text Node
\draw (366.93,277.03) node  [align=left] {{\fontfamily{ptm}\selectfont {\footnotesize (2,1)}}};
% Text Node
\draw (439.14,278.12) node  [align=left] {{\fontfamily{ptm}\selectfont {\footnotesize (3,1)}}};
% Text Node
\draw (341.99,191.49) node  [align=left] {{\fontfamily{ptm}\selectfont {\footnotesize false}}};
% Text Node
\draw (377.84,242.46) node  [align=left] {{\fontfamily{ptm}\selectfont {\footnotesize false}}};
% Text Node
\draw (396.64,192.24) node  [align=left] {{\fontfamily{ptm}\selectfont {\footnotesize true}}};
% Text Node
\draw (431.44,243.55) node  [align=left] {{\fontfamily{ptm}\selectfont {\footnotesize true}}};
\end{tikzpicture}
  \end{center}
  \caption{An example decision tree.}
  \label{fig:tree}
\vspace{-10px}  
\end{wrapfigure}
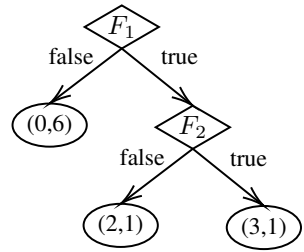
%\vspace{-10px} 

In supervised learning, a structured dataset for classification is defined as set of \emph{instances} of the form $(\vec{x},y )$ where $\vec{x} = [x_1, ..., x_n]$, $n \in \N$, is an input vector called \emph{features} and $y$ is an outcome value often called the \emph{label}. We denote by $X$ the feature space and $Y$ the outcome space. 

%\end{figure}

A \emph{decision tree} is composed of internal nodes (diamonds in Fig.~\ref{fig:tree}) and terminal nodes called leaves (ovals in Fig.~\ref{fig:tree}). Each internal node is associated with a logical formula over a feature. Each leaf node contains a set of instances, which yield a \emph{vote distribution} of the form $(n_1,\cdots,n_m)$ where $m$ is the number of classes and $n_i$ ($1\leq i \leq m$) is the number of instances of the corresponding class. For example, in Fig.~\ref{fig:tree}, the leftmost leaf node $(0,6)$ indicates that there are $0$ \texttt{class0} instances and $6$ \texttt{class1} instances. Without loss of generality, we focus on binary trees, in which internal nodes have two successors respectively called the left and right child nodes. By convention, the instances that satisfy the logical formula of an internal node go to the right child node, and those that do not satisfy go to the left child node. For example, in Fig.~\ref{fig:tree}, let $I$ be the set of training instances associated with the root node, $I_1\subset I$ be the subset that satisfies the formula $F_1$, then $I_1$ will be the set of instances associated with the right child node (with formula $F_2$),  and $I_2 = I \setminus I_1$ will be the set of instances associated with the left child (leaf) node.

Given a decision tree, any input vector (or instance) is associated with a single leaf. A decision tree is, therefore, a compact representation of a function of the form $t: X \to \N^m$, where $m$ is the number of classes. The output of a decision tree is a distribution of \emph{votes} for each class. To obtain an outcome in $Y$, we take the class with the most votes.

% A \emph{logical formula} is defined as an extension of propositional logic with arithmetic terms and comparison operators.
% The semantics of the logical language follows that of standard arithmetic and propositional logic. An \emph{arithmetic term} $T$ is defined below where $c$ is a constant (discrete or continuous value) and $var$ is a variable corresponding to (the name of) a feature:
% \begin{equation*}
% T ~:=~ c \mid var \mid -T \mid sqrt(T) \mid T + T \mid T - T \mid T * T \mid T / T
% \end{equation*}
% A \emph{Boolean formula} $F$ takes the following form where $C$ denotes a set of constants and $\oplus$ is the exclusive or operator:
% \begin{equation*}
% \begin{aligned}
% F ~:=~ &\top \mid \bot \mid var \in C \mid T < T \mid T \leq T \mid T = T \mid T > T \mid T \geq T \mid \\
%      & \lnot F \mid F \land F \mid F \lor F \mid F \limp F \mid F \oplus F
% \end{aligned}
% \end{equation*}

% In Silas~\cite{Bride2019} and the remainder of this paper, we consider two forms of logical formulae over features: 
% \begin{itemize}
%     \item If a feature $v$ is a \emph{nominal feature}, then formulae over $v$ are of the form $v\in S$, where $S$ is a set of discrete values of the feature.
%     \item If a feature $v'$ is a \emph{numeric feature}, then formulae over $v'$ are of the form $v' \geq l$, where $l$ is a numeric value of a lower bound for $v'$.
% \end{itemize}
% {\color{red}delete}

\subsection{Ensemble of Decision Trees}
\label{subsec:ensemble}

We adopt the definitions of Cui et al.~\cite{Cui2015}. Let an ensemble be a set of decision trees of size $T$. It gives the weighted sum of the trees as follows:
\begin{equation}
E(x) = \sum_{i=1}^{T} w_i\cdot t_i(x)
\end{equation}
where $E$ is the function for the ensemble, $w_i$ and $t_i$ are respectively the weight and function for each tree. The summation aggregates the weighted votes from each tree and obtains the final votes for each class. Thus, the ensemble is also a function of the signature $E: X \to \N^m$ and requires a voting mechanism to obtain the outcome. We give some famous examples of ensemble trees below.

\paragraph{Bagging.} Each decision tree is trained using a subset of the dataset that is sampled uniformly with replacement. The remaining instances form the out-of-bag (OOB) set. When selecting the best formula at each decision node in a tree, only a subset of the features are considered. This is commonly found in algorithms such as Random Forest~\cite{Breiman2001}. Bagging grows large trees with low bias and the ensemble reduces variance.

\paragraph{Boosting.} Boosting trains weak learners, i.e., small trees, iteratively as follows:
\begin{equation}
E_{i+1}(x) = E_i(x) + \alpha_i\cdot t_i(x)
\end{equation}
where $t_i$ is the weak leaner trained at iteration $i$ and $\alpha_i$ is its weight. The final ensemble is thus a special case of $E(x)$ above where $w_i$ is $\alpha_i$. The ensemble reduces bias.

AdaBoost~\cite{Freund1999} focuses on training instances that are misclassified in the previous iteration by optimising $\alpha_i$ and $t_i$ in the formula below:
\begin{equation}
\minimise_{\alpha_i,t_i} \sum_{j=0}^N L(y^{(j)}, E_i(x^{(j)}) + \alpha_i\cdot t_i(x^{(j)}))
\end{equation}
where $L$ is a loss function measuring the difference between the actual outcome $y^{(j)}$ of instance $j$ and $E_{i+1}(x^{(j)})$. AdaBoost often uses exponential loss $L(a,b) = e^{-a\cdot b}$ in which case the weak learners are trained by weighted instances.

Gradient Boosting~\cite{Friedman2000} is a generalisation of the above minimisation where each tree is trained using
\begin{equation}
    t_i(x^{(j)}) \approx - \frac{\partial L(y^{(j)}, E_i(x^{(j)}))}{\partial E_i(x^{(j)})}
\end{equation}
which is equivalent to training a regression tree using original data points but with new outcome values defined by the negative gradient. In this case, $\alpha_i$ is the learning rate.

In this work, we evaluate our approach using random forest, but our approach can also be adapted to handle boosting models.
\section{The Proposed Method: {\framename} and {\proname}}
\label{sec:method}

\begin{figure}
    \centering
    \vspace{-30pt}
    \includegraphics[width=\textwidth]{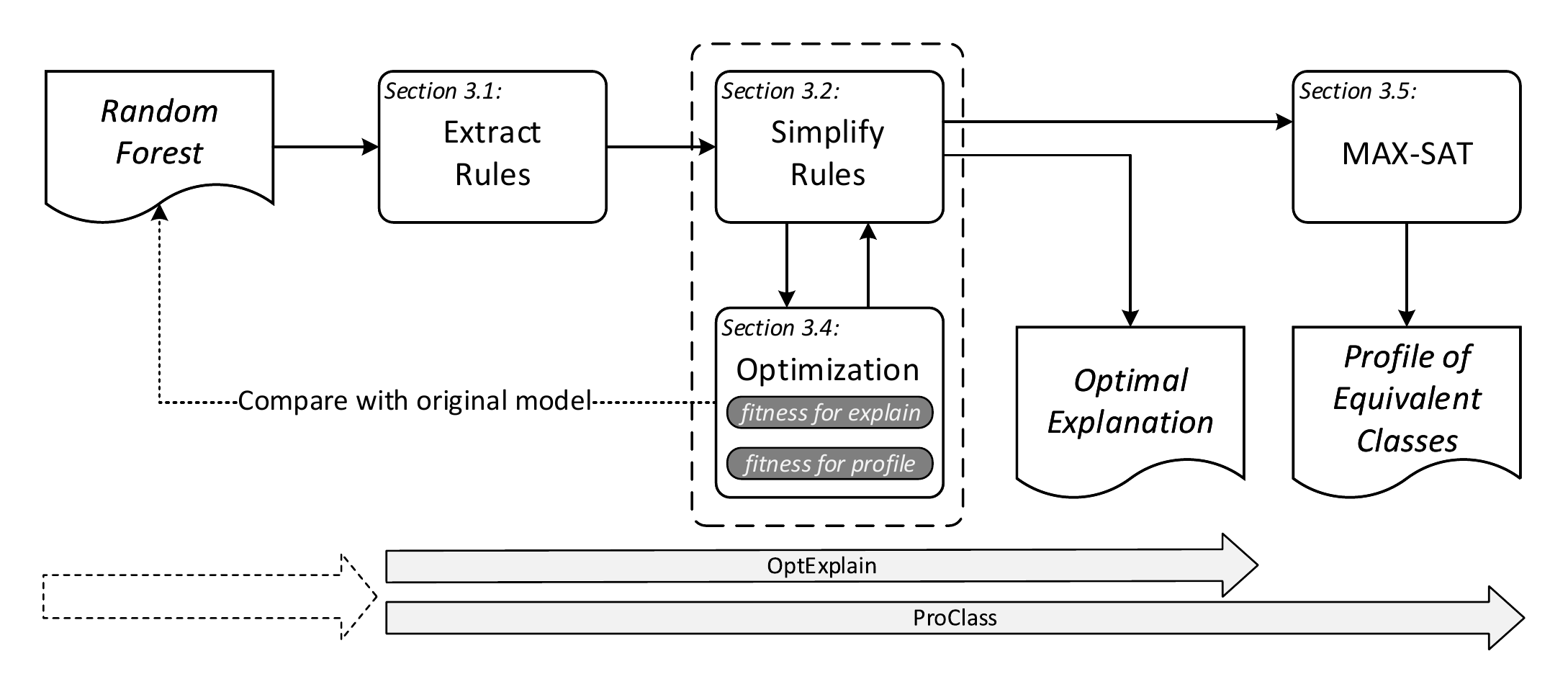}
    \caption{An overview of the proposed method.}
    \label{fig:overview}
\end{figure}

This section details our ensemble trees explanation extraction approach. We give an overview of our approach in Figure~\ref{fig:overview}. First, we get a set of decision rules by traversing the random forest. This set of decision rules is made up of branch formulae for all the trees in the forest. The size of this set is usually large, and we need to simplify the set in order to obtain a smaller set that is as close as possible to the prediction ability of the original model. We keep the nodes and branches with good quality to get the final explanation set. The preciseness of the simplification process and the scale of the explanation are related to the input parameters, and we get the optimized parameters through the bionic optimization algorithm PSO (Particle Swarm Optimization). There is two fitness function in optimization process, one dedicate to generate optimal explanations while the other one is to generate explanations used by profile.

%\subsection{Instance Level Explanation}
%\label{sec:instance_level_explanation}

%In practice, people may want to know explanations for single instances rather than the whole model. For example, after a predictor gives a high probability of raining, a user may just concern the reasons of raining, but does not concern the logic of the whole predictor. In this case, it is more appropriate to give an explanation just for the instance.

%Given a set of decision trees, the following steps are used to explain a single instance:
%1 Taking the instance as input, each decision tree is used to make a prediction. Decision paths in all the trees are collected.
%2 All predicates in the decision paths are extracted.
%3 By analysing the frequencies of the extracted predicates, we can know the importance of each predicate, so that we can explain the instance using predicates with high importance.

\subsection{From Decision Trees to Decision Rules}
\label{sec:convert}

Given a decision tree defined in Section~\ref{sec:pre}, it is straightforward to obtain the formula that is associated with each internal node. From there, we can obtain the ``decision logic'' of a branch via the \emph{branch formula} of the following form:
\begin{equation}
\label{eq:branch-fm}
(\bigwedge F) \limp D
\end{equation}
where $\bigwedge F$ is the conjunction of all the (possibly negated) internal node formula on the branch, and $D$ is the vote distribution at the leaf node. That is, if the branch goes through a node $F$ to the right branch, then we include $F$ in the conjunction, otherwise we include $\lnot F$. We refer to a formula of the above form as a \emph{decision rule}. 

% The \emph{branch conversion} algorithm runs as follows:
% \begin{itemize}
%     \item For each node of the branch with an associated formula $F:$
%     \begin{itemize}
%         \item if the branch takes the right child node, then add $F$ in the conjunction;
%         \item if the branch takes the left child node, then add $\lnot F$ in the conjunction.
%     \end{itemize}
%     \item The vote distribution $D$ can be encoded as $\bigwedge (c_i = n_i), 1\leq i \leq m$, where $c_i$ is the variable for the $i$th class and $n_i$ is the number of instances of that class in the leaf of the branch. 
% \end{itemize}

% The \emph{tree conversion} algorithm is given in Algorithm~\ref{alg:tree-convert}.

% \begin{algorithm}[H]
% \SetAlgoLined
% \KwData{a decision tree $t$}
% \KwResult{a set $R_t$ of branch formulae}
% $R_t \leftarrow \{\}$\;
% \ForEach{branch $b$ in $t$}{
%     $C \leftarrow \top$\;
%     \ForEach{node on $b$ with an associated formula $F$}{
%         \eIf{$b$ goes through the right child}{
%             $C \leftarrow C \land F$\;
%         }{
%             $C \leftarrow C \land (\lnot F)$\;
%         }
%     }
%     $D \leftarrow \bigwedge (c_i = n_i), 1\leq i \leq m$\; \tcc{$c_i$ is the variable for the $i$th class and $n_i$ is the number of instances of that class in the leaf of $b$} 
%     $R_t\leftarrow R_t \cup \{C \limp D\}$\;
% }
%  \caption{The tree conversion algorithm.}
%  \label{alg:tree-convert}
% \end{algorithm}

It is worth noting that, by the construction of decision trees, the branches are \emph{exclusive} to each other. That is, a data instance can never satisfy multiple branch formulae from the same tree at the same time. 
As a result, a decision tree $t$ can be converted into a set of mutually exclusive decision rules, denoted as $R_t$. $R_t$ can be used in classification tasks by finding the decision rule that satisfies an instance and outputting the class of the largest number of votes.

The above method can be extended to handle an ensemble of trees produced by random forest~\cite{Breiman2001} or boosting~\cite{Freund1999,Friedman2000}. In such cases, we need to consider a set $E$ of trees, and we need to multiply the vote distribution $D$ of each tree by its weight in the ensemble. The result, which we refer to as $R_E$, is the union of the set of decision rules from each tree. 

%The \emph{ensemble trees conversion} algorithm is given in Algorithm~\ref{alg:ensemble-convert}.

% \begin{itemize}
%     \item For each tree $t_i\in E$ which has weight $w_i$, where $1\leq i \leq T$ and $T$ is the number of trees in $E$: 
%     \begin{itemize}
%         \item convert $t_i$ to a set $R_{t_i}$ of decision rules;
%         \item For each vote distribution $D$ in $R_{t_i}$, multiply the vote count of each class by the weight $w_i$.
%     \end{itemize}
%     \item The set $R_E$ of decision rules for the ensemble $E$ is formed by $R_E = \bigcup R_{t_i}, 1\leq i \leq T$.
% \end{itemize}

% \begin{algorithm}[H]
% \SetAlgoLined
% \KwData{a set $E$ of weighted decision trees}
% \KwResult{a set $R_E$ of decision rules}
% $R_E \leftarrow \{\}$\;
% \ForEach{tree $t_i\in E$ which has weight $w_i$}{
%     obtain $R_{t_i}$ via Algorithm~\ref{alg:tree-convert}\;
%     \ForEach{vote distribution $D$ in $R_{t_i}$}{
%         multiply the vote count of each class by $w_i$\;
%     }
%     $R_E\leftarrow R_E \cup R_{t_i}$\;
% }
%  \caption{The ensemble trees conversion algorithm.}
%  \label{alg:ensemble-convert}
% \end{algorithm}

Unlike the set of decision rules for a single tree, that for an ensemble of trees may contain decision rules that are \emph{not} exclusive to each other. In fact, for an ensemble $E$ of $n$ trees and for any data instance, there should be exactly $n$ decision rules in the set $R_E$ 
%computed by Algorithm~\ref{alg:ensemble-convert} 
that are satisfied by the instance --- one from each tree. To use $R_E$ in classification tasks, one can find the subset of decision rules that are satisfied by an instance $x$ and aggregate the weighted vote counts for each class from those rules. The class with the highest weighted count is the output, which we refer to as $R_E(x)$.
%We denote by $R_E(x)$ the aggregated vote distribution.  
The following result is straightforward:

\begin{proposition}
\label{prop:equiv}
For any ensemble trees model $E$ and any data instance $x$, suppose $R_E$ is the set of decision rules of $E$ derived by the method above,
%computed by Algorithm~\ref{alg:ensemble-convert}, 
then $E(x) = R_E(x)$.
\end{proposition}

In the sequel, we shall denote the original ensemble trees model as $E$
and the converted set of decision rules as $R_E$.
%Algorithm~\ref{alg:ensemble-convert} 
% the above conversion algorithm
% as $\conv$, and write the conversion as $\conv(E) = R_E$ where $R_E$ denotes the set of decision rules.

\subsection{Simplification of Decision Rules}
\label{sec:simp}

As discussed in Section~\ref{sec:intro}, some ensemble trees models used in real-life applications are enormous and complex. Consequently, the converted set of decision rules for such a model consists of a huge number of rules and each rule may be a very long formula. To reduce the complexity of the explanation, we consider simplifying the set of decision rules in two dimensions: the length of the rules and the number of the rules.

% \subsubsection{Logically equivalent simplification}
% \label{sec:equiv-simp}

Continuing from the output $R_E$ of %Algorithm~\ref{alg:ensemble-convert}, 
Section~\ref{sec:convert},
each formula in $R_E$ is a branch formula, which we can simplify using a node filter as step one.

% The literature on explaining decision trees~\cite{perner2011} often argues that the top nodes in a tree are very important, as they are selected by testing a large subset of data points. One may also argue that the nodes close to the leaves are also important, as they dictate the final classification. As such, we propose a node weight multiplier that reduces the weight of nodes in the middle of a tree branch. Since the shape of such a function is similar to a quadratic function, we use the following parameters for the weight multiplier:

%\begin{param}[Node Weight Multiplier $\nodeweightmult$]
%\label{param:node-weight-mult}
%The node weight multiplier is a quadratic function 
%\begin{equation}
%    \nodeweightmult = a \times r^2 + b \times r + c
%\end{equation}
%where $a$, $b$ and $c$ are parameters of the multiplier and $r$ is the rank of the node calculated as below.
%\begin{equation}
%    r = \frac{d}{L}
%\end{equation}
%where $d$ is the depth of the node from the root and $L$ is the number of nodes on the branch.
%\end{param}

%\zhe{Intuition: $\delta = 4r^2 - 4r + 1$ is the simplest case where the top nodes' multipliers are near $1$, the bottom nodes' multipliers are also near 1, while the middle nodes' multipliers are near $0$, which fit the above idea. Obviously we can optimise the values of $a,b,c$. We can also find better functions that suit this step.}

\begin{param}[Node Filter $\nodefilter$]
\label{param:node-filter}
We measure the “quality” of a node ($NQ$) by information gain (IG):
\begin{equation}
    \label{equ:NQ}
    NQ(n)=IG(n)
\end{equation}
where $n$ is the target node. We scan the nodes of each branch in each tree, and remove the nodes with $NQ$ below $\nodefilter$, which is a real positive number.

\end{param}

The second step is to simplify each branch formula by merging the nodes.

\begin{lemma}
\label{lem:simp-conjunct}
For any branch formula $(\bigwedge F_i) \limp D$, where each $F_i$ is the formula associated with a node on the branch, the left-hand side can be simplified to conjunction normal form (CNF) with at most $n$ conjuncts, where $n$ is the number of features of the dataset.
\end{lemma}

\begin{proof}
By the construction of decision trees, for any two conjuncts $v\geq l$ and $v\geq l'$ over a numeric feature $v$ that appears in $\bigwedge F_i$, we can simplify them into one conjunct $v\geq l''$ where $l'' = max(l,l')$. 

For any two conjuncts $v\geq l$ and $\lnot (v\geq l')$ over $v$, if $l < l'$, then they can be simplified to $l \leq v < l'$. Otherwise, we have $l\geq l'$. In this case, the two conjuncts do not have an intersection, and there would be no instances at the leaf node and the training algorithm should not let this case happen.

For any two conjuncts $\lnot (v \geq l)$ and $\lnot (v \geq l')$ over $v$, they can be simplified to $v < l''$ where $l'' = min(l,l')$.

For any two conjuncts $v'\in C$ and $v'\in C'$ over a nominal feature $v'$, they can be simplified into one conjunct $v'\in C''$ where $C'' = C \cap C'$. If any of the two conjuncts is negated, e.g., $\lnot (v'\in C)$, we can simply take the complement set $C''' = C_v \setminus C$, where $C_v$ is the full set of permitted discrete values of the feature $v$, and convert the negated conjunct into $v' \in C'''$. The remainder of the proof is analogous. 

Thus, the left-hand side of the branch formula can be simplified to one conjunct per feature.
\qed
\end{proof}

Note that the simplification of Lemma~\ref{lem:simp-conjunct} preserves logical equivalence of the set of decision rules while the other steps of this section do not. One may ask why do we use the node filter when Lemma~\ref{lem:simp-conjunct} can merge the nodes. The reason is that using all the nodes in a branch may result in very narrow and focused explanations, so we use $\theta$ as a parameter to adjust the scope.

% Up to this point, the simplified set $R_E'$ of decision rules is still logically equivalent to the original ensemble trees model. While this property is crucial in verification tasks, it is not essential for explanation purposes. In the below steps, we will simplify $R_E'$ further as in a ``lossy compression''.

% \subsubsection{Node filtering}
% \label{sec:filter}

The above two steps aim to shorten the decision rules. The third step is to reduce the number of rules by filtering out those of low quality.

\begin{param}[Rule Filter $\rulefilter$]
\label{param:rule-filter}
We measure the ``quality'' of a decision rule ($RQ$) by the following formula:
\begin{equation}
\label{equ:RQ}
    RQ(r)=\frac{(log_2(m) - H(l_r))}{log_2(m)} \times Acc
\end{equation}
where $r$ is the target rule, $m$ is the number of classes, $H(l_r)$ is the entropy of the leaf of the rule, and $Acc$ is the accuracy of the corresponding tree on the OOB dataset. We remove a rule if its $RQ$ is less than $\rulefilter$, which is a real number between $0$ and $1$.
%{\color{red} Z: e.g., can use (multi-class) AUC rather than Acc. Or other measures. Need experiment. Or change the formula completely.}
\end{param}

The fourth step merges decision rules into \emph{groups of the same class signature}.

\begin{param}[Leaf Merger $\leafmerger$]
\label{param:leaf-merger}
Given a vote distribution $(n_1,\cdots,n_m)$, where $m$ is the number of classes, we convert the distribution into ratios $(\xi_1$,$\cdots$,$\xi_m)$ where each $\xi_i, 1\leq i \leq m$, is the ratio of class $i$ in the leaf node. The class signature of this leaf node is defined as the tuple $(\ceil{\xi_i/\leafmerger},\cdots,\ceil{\xi_m/\leafmerger})$, where $\leafmerger$ is a real number between $0$ and $1$.
\end{param}

% {\color{red} Z: we can also use a data-oriented approach. For example, we can use the training data (or OOB data or testing data) to test which decision rules are often satisfied by the same data instance. If a set of rules are satisfied at the same time by at least 10 (this can be a parameter) instances, and they have the same majority class, then we put them into the same group.}

Using the above definition, we divide the set of decision rules into a set of sets $\{G_1,\cdots,G_j\}$. Each $G_i$, $1\leq i \leq j$, contains the set of rules of the same class signature. 
Intuitively, a larger $\leafmerger$ yields fewer distinct equivalent classes/groups and vice versa. We use this parameter to control the number of groups in the final explanation.
%and has the same majority class. 
% More specifically, suppose $G_i$ contains the following decision rules, and $S_i$ is the class signature of $G_i$:  
% \begin{align*}
% r_1:F_1^i \land \cdots & \land F_{m_1}^i \limp S_i\\
% \vdots & \\
% r_n:F_{n}^i \land \cdots & \land F_{m_n}^i \limp S_i.
% \end{align*}

%$$(F_{n}^i \land \cdots \land F_{m_n}^i) \land \cdots \land (F_{n}^i \land \cdots \land F_{m_n}^i).$$
%We convert the above formula into CNF, and obtain an equivalent formula of the form 
%$$C_1^i\land \cdots \land C_k^i.$$
%Then, we compute the weight of each clause $C_j^i$, for $1\leq j \leq k$, by testing the training dataset using  $C_j^i$ and calculating the resulting cross-entropy. 

% Specifically, we collect all the conjuncts in $G_i$, and associate the following weight to each conjunct:
% \begin{equation}
%     IG \times (log_2(m) - H(l)) \times Acc
% \end{equation}
% where $IG$ is the information gain of the node corresponding to the conjunct, the other notations are the same as those in Parameter~\ref{param:rule-filter}. 

%{\color{red} Z: again, can use (multi-class) AUC rather than Acc. Or other measures. Need experiment. Also, instead of information gain (IG), we can try \underline{gini impurity} or \underline{information gain ratio} etc.}

For a large-scale random forest, the filtered rules are still a large-scale formula set. In the last step, we control the number of decision rules we get.
\begin{param}[Size Filter $\sizefilter$]
\label{param:size-filter}
In each $G_i$, we take the number of instances in the leaf node of each rule as the \textbf{weight of the rule}, and we select $\sizefilter$ rules in a weight-first manner. %\zhe{What's the ``value of the leaf node''?}
\end{param}

Associating each rule with the weight is crucial because now the vote distribution has been converted into ratios, and we have lost the information of the number of votes in the distribution. The weight retains this information. For example, we would prefer a ratio of $(0.7,0.3)$ with 100 votes to overwhelm a ratio of $(0.1,0.9)$ with 10 votes.

% Now we obtain a smaller set of shorter decision rules, in which each rule's right-hand side is no longer a vote distribution but a class signature, which can also be used to find the majority class as output. 
To summarise, we denote the composition of the above steps as $\simp$, and write the simplification procedure as follows:
\begin{equation}
\label{equ:simp}
    \simp(R_E,\nodefilter,\rulefilter,\leafmerger,\sizefilter) = R_E'
\end{equation}
where $R_E'$ is a set of sets of decision rules where each rule has a weight. We use $R_E'$ as our explanation to the original model.
%, which can serve as an explanation of the original ensemble trees model.

%{\color{red} Z: what other steps and parameters can we include in this procedure?}

%\hadrien{Can consider abstracting/discretising the value of features if performance needs to be improved.}

% \subsection{Evaluation of the Explanation}
% \label{sec:eval}

\subsection{Prediction}
\label{sec:prediction}

To use $R_E'$ in a classification task, let $x$ be a data instance, we first find all the decision rules in $R_E'$ that are satisfied by $x$, then multiply the class signature of each rule by the corresponding weight, and finally add up to get a tuple. The class with the largest value is the output. This procedure is denoted as $R_E'(x) = c$ where $c$ is a class. We give an example in Fig. \ref{fig:prediction} where we put the satisfied rules in red boxes.

\begin{figure}[t!]
    \centering
    \vspace{-10px}
    \includegraphics[width=12cm]{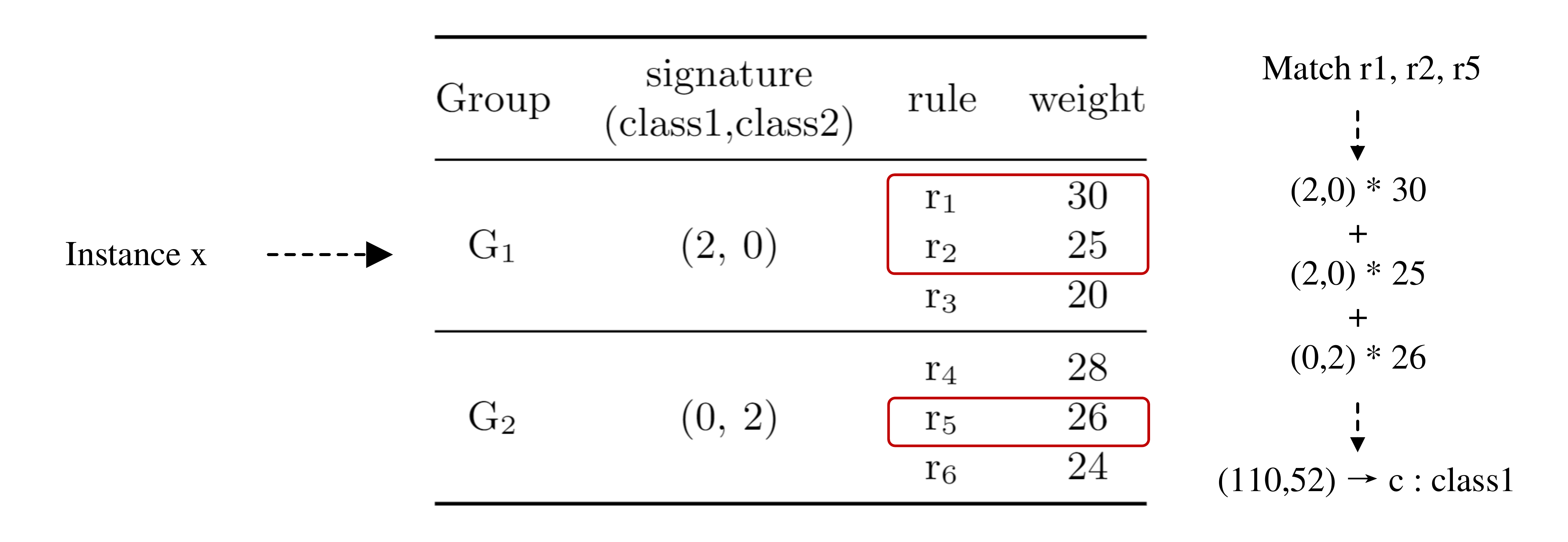}
    \caption{An example of the prediction procedure using the simplified rules $R_E'$.}
    \label{fig:prediction}
\end{figure}

\subsection{Optimal Explanations}
\label{sec:opt}

%Most existing work simply pick an explanation that makes sense. Here 
Now we consider a step further: how to evaluate explanations and find the optimal one? Intuitively, a good explanation should meet the following two criteria: 
\begin{itemize}
    \item The classification behaviour of the explanation $R_E'$ should be similar to the original model $E$.
    \item The explanation should be concise and small.
\end{itemize}

We use \textit{fidelity} to measure the first criterion. Fidelity is defined as the degree of similarity between the predictions of $R_E'$ and $E$ on unseen data~\cite{papenmeier2019}. First, we take the test set without labels as $D$. Then we use the classification results on $D$ from the original model $E$ as the ``ground truth''. Lastly, we evaluate the classification accuracy of the explanation $R_E'$ on $D$ as fidelity. The fidelity component is denoted as $\fstmeasure(R_E',E,D)$.

The second criterion is \textit{scale} which is measured by the total number of conjuncts in the rules of $R_E'$ and is denoted as $\sndmeasure(R_E')$.

The \emph{score} of an explanation $R_E'$ is defined as 
\begin{equation}
\label{equ:optscore}
    \scoreOpt(R_E') = {\fstmeasure(R_E',E,D) \times \accweight}+\frac{1-\accweight}{1+e^{5\times(\frac{\sndmeasure(R_E')}{m \times n}-1)}}
\end{equation}
where $\accweight$ is a real number between 0 and 1, $n$ is the number of features, and $m$ is the number of classes. 
%For different users and datasets, the expected explanations are often different. 
The intuition is that the score grows linearly with fidelity, but it drops significantly when the explanation is too large. The second component is a sigmoid-shaped function. 
Also, a large $\accweight$ puts more importance on fidelity, and a small $\accweight$ puts more importance on scale. 
% means that users attach importance to the fidelity of the explanation. Relatively, a smaller $\accweight$ means that users need a smaller and more understandable explanation.

\begin{figure}[htbp]
\centering
\begin{floatrow}
    \begin{subfigure}[The effect of $\nodefilter$ and $\rulefilter$ on \textit{scale} coverage. (The original scale is 65247.)]{
        \begin{minipage}[t]{0.5\linewidth}
        \centering
        \includegraphics[width=6.3cm]{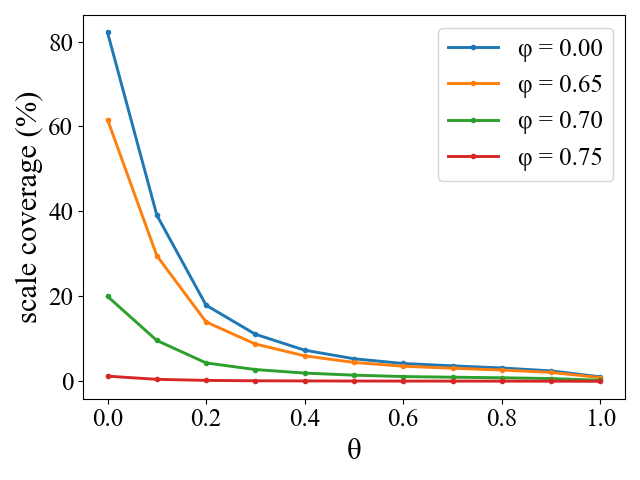}
        %\caption{fig1}
        \label{fig:theta-phi}
        \end{minipage}}
    \end{subfigure}\hfill
    \begin{subfigure}[The effect of $\sizefilter$ on \textit{fidelity}.]{
        \begin{minipage}[t]{0.5\linewidth}
        \centering
        \includegraphics[width=6.3cm]{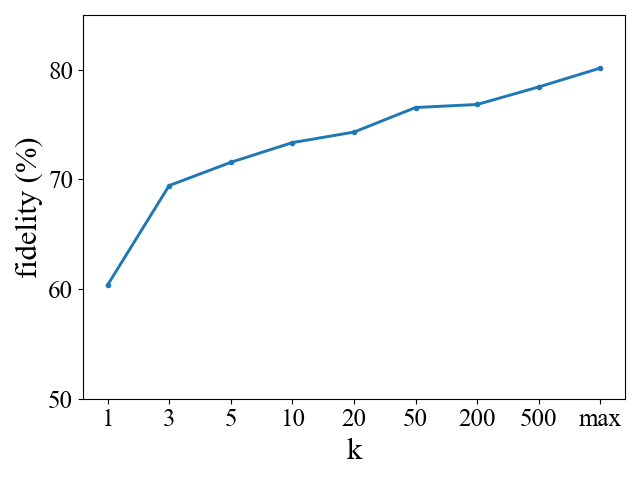}
        %\caption{fig1}
        \label{fig:k-acc}
        \end{minipage}}
    \end{subfigure}
%\centering
\caption{The effect of parameters on the two criterion.}
\end{floatrow}
\end{figure}

The parameters in Section~\ref{sec:simp} introduce a great variety of potential explanations. Particularly, the two parameters $\nodefilter$ and $\rulefilter$ control the strictness of the node filter and the rule filter respectively. We experiment on the diabetes dataset~\cite{Dua:2019}
%from UCI Machine Learning Repository~\cite{Dua:2019} 
to explore the effect of parameters on the two criteria. Fig. \ref{fig:theta-phi} shows the effect of the $\nodefilter$ and $\rulefilter$ on scale. As $\nodefilter$ and $\rulefilter$ rise, we get significantly smaller explanations. Even when $\nodefilter$ and $\rulefilter$ are both 0, the scale has dropped by about 20\% compare to the original model due to Lemma \ref{lem:simp-conjunct}. We also explore the effect of $\sizefilter$ on fidelity. We set $\nodefilter$ to 0.65 and $\rulefilter$ to 0.7. Fig. \ref{fig:k-acc} shows the average of 50 tests. 
%From the results, we conclude 
The results show
that a relatively small value of $\sizefilter$ can get a fidelity close to all the rules.

To obtain an optimal explanation, we use the \emph{linearly decreasing inertia weigh particle swarm optimization} algorithm (LDIW-PSO) \cite{785511} to optimise the parameters mentioned above with the $\scoreOpt$ as the fitness, which is the objective function to be optimised. Then we apply the optimal parameters to \textit{\simp} and obtain the optimal explanation $R_{opt}$. We refer to the above procedure as \emph{\framename}.
% %We write the optimisation procedure below.
% %as $\framename$:
% \begin{equation}
% \label{equ:opt}
%     \framename(E,D,\accweight) = R_{opt}
% \end{equation}

%Through experiments, we found that the same parameters used on different random forest have different effects, even for models trained on the same data set. For example, the data set of diabetes, which has only 768 samples, each time a different training set is divided to obtain different random forests, and the same parameters obtain completely different results in these random forests.  

\subsection{Profile of Equivalent Classes}
\label{sec:pro}

%How to understand a random forest more intuitively? 
An explanation with high fidelity usually has large scale, while a concise explanation can allow users to quickly understand the predictive behavior of the model at the sacrifice of fidelity. Sometimes the latter is preferred to draw a high-level conclusion of the classification behavior. 
%users do not care about the prediction accuracy of the explanation, but rather want to understand the conclusions drawn from a random forest, that is the precise description of the class. 
We propose a new method called the \textit{Profile of Equivalent Classes (\proname)}, which is a more concise description of classes based on the extracted decision rules.

The \textit{\proname} process is based on the $R_E'$ obtained by Equation~\ref{equ:simp}. The $R_E'$ may contain several groups, each has no more than $\sizefilter$ rules. Then we merge the decision rules in each group by solving a weighted maximum satisfiable (MAX-SAT) problem to get a new explanation $R_E''$. Different from Section \ref{sec:opt}, \textit{\proname} requires more rules as input, and the number of groups is equal to the number of classes. The number of groups in $R_E''$ is denoted as $\thdmeasure(R_E'')$. For the above needs, we defined another \textit{score} for the optimization in \emph{\proname}:
\begin{equation}
\label{equ:proscore}
    \scorePro(R_E'') = (m-\thdmeasure(R_E'')+1)\times\sndmeasure(R_E'')
\end{equation}
where $m$ is the number of classes. Using $\scorePro$ as fitness, we obtain the optimized result is denoted as $R_{opt}'$. An ideal $R_{opt}'$ has $m$ groups $\{G_1,\cdots,G_m\}$. 
% Since the decision rules in each group $G_i$ ($1\leq i \leq m$) have the same class signature, we can focus on the left-hand side of rules, and conjunct all the rules within a group. 
For each group $G_i$ ($1\leq i \leq m$), we associate the weight of each rule to each of its conjuncts, and send all the weighed conjuncts in the group to a SAT solver such as Z3~\cite{deMoura2008}. 
% in a conjunction, which can be written as
% $$r_{1}^i \land \cdots \land r_{n}^i.$$
%A SAT solver such as Z3~\cite{deMoura2008} 
The solver
will return a subset of satisfied conjuncts that maximise the total weight. After solving the MAX-SAT problem for each $G_i$, we obtain a subset of conjuncts. Then we simplify the conjuncts into one rule $r_i'$ using Lemma~\ref{lem:simp-conjunct}. Performing the above steps on all groups, we get the profile of equivalent classes denoted as $R_{pro}$, which has the form
$$R_{pro} = \{r_1' \limp S_1, \cdots , r_m' \limp S_m\}.$$
where $r_i'$ is the logic for predicting the class $S_i$.

% To summarise, we denote the composition of the above steps as $\proname$, and write the procedure as follows:
% \begin{equation}
% \label{equ:pro}
%     \proname(E) = R_{pro}
% \end{equation}

\section{Case Studies and Experiment}
\label{sec:case_studies}

In this section, we demonstrate our method through case studies. We used 
%RandomForestClassifier in 
scikit-learn to train random forest models. We implemented our method in Python and evaluated it on multiple datasets: adult, credit, diabetes, german, mnist, spambase, all of which are available on OpenML~\cite{openml}. %In particular, we show the $R_{opt}$ and $R_{pro}$ on diabetes dataset.
There are two important parameters in LDIW-PSO algorithm: particle size and iteration period. In our experiments, both particle size and iteration period are set to 20 by default.
Experiments were conducted on a machine with an Intel Core i9-7960X CPU and 32GB RAM.

\subsection{Case Study 1: Diabetes Prediction}

We first evaluate \emph{\framename} on diabetes dataset~\cite{Dua:2019}, which has 8 features, 2 classes, and 768 samples. The eight features are the number of times pregnant (preg), plasma glucose concentration (plas), diastolic blood pressure (pres), 2-hour serum insulin (insu), triceps skinfold thickness (skin), body mass index (mass), diabetes pedigree function (pedi) and age. We randomly select 100 samples as the testing set, and the rest as the training set. Then we train a random forest with 100 trees and unlimited depth.
%which can predict whether a instance is diabetes negative or positive.

\begin{table}[h]
\caption{Optimized parameters and predictive performance of the explanation.}
\setlength{\tabcolsep}{1.5mm}
\begin{tabular}{@{}ccclccc@{}}
\toprule
\multirow{2}{*}{$(\phi,\theta,\psi,k)$} & \multicolumn{2}{c}{$R_E$} & & \multicolumn{3}{c}{$R_{opt}$}  \\ \cmidrule(lr){2-3} \cmidrule(l){5-7} 
                           & scale    & accuracy &  & scale     & accuracy    &  fidelity                            \\ \midrule
(0.55, 0.45, 0.83, 3.0)                  & 102584     & 80\%  &      & 6         & 80\%         & 92\%                          \\ \bottomrule
\end{tabular}
\label{tab:optimising}
\end{table}

\begin{table}[]
\caption{The optimal explanation $R_{opt}$.}
\setlength{\tabcolsep}{5mm}{
\begin{tabular}{|c|c|l|c|}
\hline
Groups     & \begin{tabular}[c]{@{}c@{}}Class Signature\\ (negative,positive)\end{tabular} & \multicolumn{1}{c|}{Rules}                           & Weight \\ \hline
\multirow{2}{*}{Group1} & \multirow{3}{*}{(2.0 , 0.0)}  & $pedi \leq 0.7$               & 30 \\ \cline{3-4} 
                        &                               & $plas \leq 130.0$             & 23 \\ \cline{3-4} 
                        &                               & $plas \leq 157.5$             & 21 \\ \hline
\multirow{2}{*}{Group2} & \multirow{3}{*}{(0.0 , 2.0)}  & $mass > 28.7$                 & 30 \\ \cline{3-4} 
                        &                               & $age > 27.5$                  & 22 \\ \cline{3-4}
                        &                               & $plas > 122.5$                & 20 \\ \hline
\end{tabular}}
\label{tab:diabetes OPT}
\end{table}

We set $\accweight$ to 0.9, and \emph{\framename} will produce an explanation $R_{opt}$. The result is shown in Table \ref{tab:optimising} and Table \ref{tab:diabetes OPT}. 
%According to $Proposition$ 1, the prediction of $R_E$ is the same as the original model. 
Recall that $R_E$ is equivalent to the original model $E$ (Proposition~\ref{prop:equiv}).
Both $R_E$ and $R_{opt}$ have 80\% accuracy on the test set. $R_E$ has 11106 rules with 102584 conjuncts, while $R_{opt}$ has 6 rules with 6 conjuncts. The fidelity of $R_{opt}$ is 92\%,
which means it is very similar to the original model.
%which means that 
%92\% of $R_{opt}$'s prediction is the same as $R_E$'s prediction. Users 
%the user can make a prediction using $R_{opt}$, and the result is likely the same as the prediction of $R_E$. 
By observing the decision rules, users can analyze what role each feature plays in the prediction process. In this explanation, if an instance has $plas > 157.5$, 
%a plas feature greater than 157.5, 
then the chance of prediction positive diabetes is high. 
%there is a high probability of positive prediction.

\begin{table}[]
\caption{The profile of equivalent classes of a random forest model for diabetes.}
\setlength{\tabcolsep}{1.5mm}
\begin{tabular}{@{}ccccc@{}}
\toprule
&\textbf{negative}      &  & \textbf{positive}&          \\ \cmidrule(r){2-2} \cmidrule(l){4-4}
&$2.0<preg\leq2.5$ &  & $7.5<preg\leq8$&       \\
&$90<plas\leq91.5$ &  & $173.5<plas\leq175$&   \\
&$74<pres\leq74.5$ &  & $70<pres\leq71$&       \\
&$28.5<skin\leq29$ &  & $23.5<skin\leq24$&     \\
&$61.5<insu\leq63$ &  & $126.5<insu\leq127.5$& \\
&$25.9<mass\leq26$ &  & $32.9<mass\leq33$&     \\
&$pedi=0.2$        &  & $pedi=0.5$&          \\
&$23.5<age\leq24$  &  & $33<age\leq33.5$&      \\ \bottomrule
\end{tabular}
\label{tab:diabetes Pro}
\end{table}

%We also assess $\proname$ on this dataset. 
Table \ref{tab:diabetes Pro} gives the profile 
%$G_{pro}$ 
of equivalent classes derived from $\proname$. The profile is divided into two groups corresponding to equivalent class labels: negative diabetes and positive diabetes. Medical practitioners can quickly obtain the salient feature values corresponding to each class according to the profile.

\subsection{Case Study 2: Digit Recognition (MNIST)}

In order to visually demonstrate the profile, we use the MNIST dataset~\cite{openml} to illustrate $\proname$. The MNIST dataset contains 70,000 images of size 28 $\times$ 28 pixels. In order to reduce the complexity, we randomly picked 10,000 samples from the dataset as the training set, and train a random forest with 100 trees and 15 max-depth. $\proname$ produces a profile $R_{pro}$ which gives the range of some pixels. Then we take the median of the range as the pixel value. The remaining pixels that do not appear in the profile are set to light blue. 

\begin{figure}
    \centering
    \includegraphics[width=12cm]{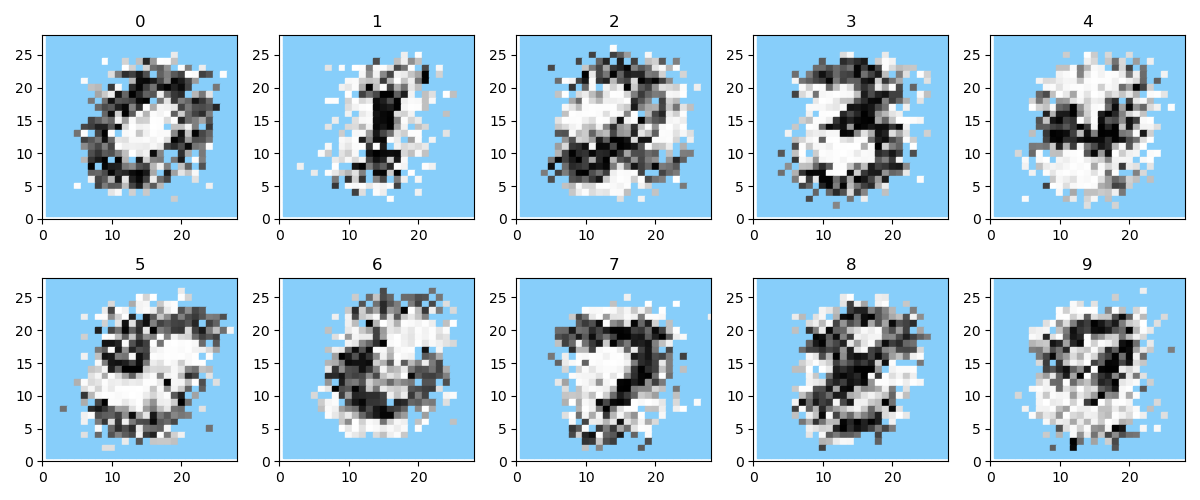}
    \caption{As visualisation of $R_{pro}$ on MNIST.}
    \label{fig:profile}
\end{figure}

From Fig.~\ref{fig:profile} we can see the profile's visualization of each digit. 
%We can recognize the number corresponding to each picture. 
For peripheral pixels that are not significant in the prediction, the profile will not give a description. As we can see, the profiles indeed yield human-understandable visualisation of each class.

\subsection{Experiment and Comparison}

We compared the proposed method to a baseline method --- Hara and Hayashi’s approach named degragTrees~\cite{hara2018}. In their work, defragTrees has been compared with BATrees~\cite{Breiman1996BornAT}, inTrees~\cite{deng2019} and Node Harvest~\cite{Meinshausen2009NodeH}. Their result suggests that defragTrees generated smaller set of rules with higher fidelity than the other methods. %The defragTrees is implemented also in Python. 
We choose the following datasets: adult, credit, diabetes, german, spambase~\cite{Dua:2019}.
We split the datasets into two subsets at random: a 70\% training set, and a 30\% testing set. Then we train the ensemble trees with 100 trees and 10 max-depth for \emph{\framename} and defragTrees. For \emph{\framename}, we set two experimental groups and set $\accweight$ as 0.5 and 0.9 respectively.

\begin{table}[]
\caption{\emph{\framename} and defragTrees}
\setlength{\tabcolsep}{1.8mm}
\begin{tabular}{@{}lccccclccl@{}}
\toprule
\multirow{3}{*}{Data} & \multicolumn{2}{l}{\multirow{2}{*}{defragTrees}} & \multicolumn{1}{l}{} & \multicolumn{5}{l}{\emph{\framename}}                           &  \\
                      & \multicolumn{2}{l}{}                             & \multicolumn{1}{l}{} & \multicolumn{2}{l}{$\accweight$=0.5} &     & \multicolumn{2}{l}{$\accweight$=0.9} &  \\ \cmidrule(lr){2-3} \cmidrule(lr){5-6} \cmidrule(lr){8-9}
         & scale & fidelity\% &  & scale & fidelity\% &  & scale & fidelity\% &  \\ \cmidrule(r){1-9}
adult    & 65.5        & 81.9   &  & \textbf{37.2}& 87.0   &  & 43.2  & \textbf{88.4}&  \\
credit   & 8.7         & 85     &  & \textbf{2.8} & 94.2   &  & 9.8   & \textbf{94.6}&  \\
diabetes & 14.2        & 74.8   &  & \textbf{4}   & 85.8   &  & 8.3   & \textbf{88.3}&  \\
german   & \textbf{5.4}& 69.8   &  & 19           & 87.3   &  & 50.2  & \textbf{88.7}&  \\
spambase & 82          & 91.7   &  & \textbf{22}  & 91.6   &  & 40.6  & \textbf{92.9}&  \\ \bottomrule
\end{tabular}
\label{tab:comparison}
\end{table}

We compare the two methods with respect to the two criteria: scale and fidelity. We conducted the experiment over ten random data realizations for each dataset. Table \ref{tab:comparison} shows the average values of the ten tests. The number in bold is the best scale/fidelity value for each dataset. Table \ref{tab:comparison} shows that the scale of explanations generated by \emph{\framename} with 0.5 $\accweight$ is the smallest except on the german dataset, and the fidelity is generally better than defragTrees. The small-scale explanation on german generated by defragTrees has much lower fidelity. \emph{\framename} with 0.9 $\accweight$ can generate the highest fidelity explanation with similar scale than defragTrees. 
%We conclude that \emph{\framename} can generate high-fidelity, small-scale explanation. At the same time, the result suggest that we may obtain a proper explanation by changing the value of $\accweight$.
In both settings of $\accweight$, \emph{\framename} generates superior explanations in most cases.

\begin{figure}
    \centering
    \includegraphics[width=12cm]{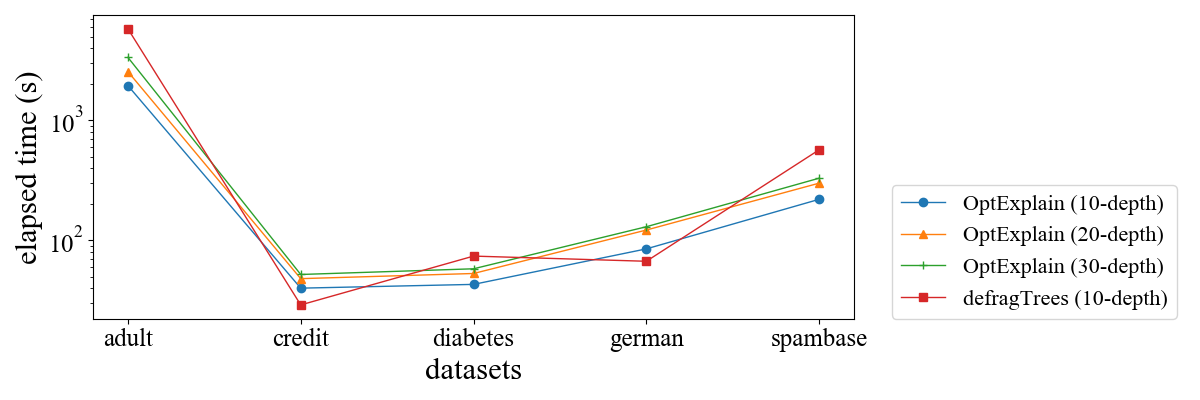}
    \caption{Mean elapsed time comperision. (The y-axis is logarithmic scale.)}
    \label{fig:time}
\end{figure}

Finally, we visually compare the computation time in Fig. \ref{fig:time}, and show the mean time of 10 tests. We run \emph{\framename} on models of 100 trees with depth of 10, 20, and 30 respectively. The method defragTrees on 100 trees and depth 10 is also used as a reference. The results show that \emph{\framename} has better computational performance than defragTrees on adult, diabetes and spambase. 
%\emph{\framename} scales very well with the depths of the trees (size of the model) and is more efficient especially on larger models. 
\emph{\framename} can also deal with very large models, and for the above three datasets, \emph{\framename} can generate explanations for trees of depth 30 faster than defragTrees can for depth 10.
For a concrete example, the mean computation time of \emph{\framename} on the adult dataset is 1,949 s, while the time of defragTrees is 5,797 s (both depth 10).

\section{Related Work}
\label{sec:related}

There are a variety of approaches for tackling machine learning interpretability. Some recent and popular ones are focused on local explanations. LIME~\cite{lime2016} is such a tool that finds linear approximations of the prediction model and gives importance weights for certain predicates used in the prediction. Anchors~\cite{ribeiro2018} generates ``if-then'' style explanations for predictions. Such explanations have similar forms as our decision rules and are considered intuitive and easy to understand by the user. Shapley (SHAP) Values~\cite{Lundberg2017} are often used to extract importance scores and impacts on features. Like LIME, SHAP also provides user-friendly graphical presentations (e.g., bar charts) for explaining predictions. It should be noted that SHAP can also be used to obtain global feature importance. The above three methods are \emph{model-agnostic}, which means that in the process of providing explanations and making machine learning more ``white-box'', they take prediction models as a ``black-box'' and attempt to find patterns of \emph{features} when the model makes predictions. An advantage is that they can be applied to different machine learning techniques, including ensemble trees and neural networks.  CHIRPS~\cite{hatwell2020} is another technique for local explanations. In contrast to the above techniques, CHIRPS looks into decision trees and uses frequent pattern mining on decision nodes to obtain decision rules as the explanation. 

% More closely related to this work are those techniques that simplify the original model into certain forms as global explanations. 
Global explanations are more closely related to this work.
Recent examples include Hara and Hayashi's approach~\cite{hara2018} that uses Bayesian model selection to extract decision rules. Deng's inTrees~\cite{deng2019} extracts, selects and prunes rules from a set of decision trees and uses frequent pattern analysis to summarise rules into a smaller prediction model. 

Most relevant techniques come from a statistics perspective. Their simplification process often consists of selection, pruning and frequency analysis. By contrast, our work comes from a logician's point of view and, on top of the usual selection and pruning operations, uses automated reasoning to simplify logical formulae and find abstractions of equivalent classes. We see many related methods as complementary ones rather than competitors because they output in different forms. For example, one can combine SHAP values and our work to form a more comprehensive explanation.

\section{Conclusion and Future Work}
\label{sec:conc}

This paper presents a streamlined procedure for extracting optimal \emph{logical} explanations from ensemble trees models. As an additional feature, our method can also output the ``profile'' of each class so that the user can \emph{see} how the model predicts in different cases. We give two case studies to illustrate how our method works and show that our method outperforms state-of-the-art through experimental results. 

As future work, we plan to use this work to perform efficient verification tasks. In particular, we will use the optimal explanation as an approximation of the original model and try to ``debug'' the model against user-specified properties. The rules that violate the property can serve as constraints to narrow down the search space when finding counterexamples. Another important future direction is to convert deep neural networks to ensemble trees and extend this work to explain deep learning. 

\bibliographystyle{splncs04}
\bibliography{main}

\end{document}